\documentclass[runningheads]{llncs}

\usepackage{latexsym}
\usepackage{amssymb}  
\usepackage{amsmath}
\usepackage{stmaryrd}
\usepackage{xy}
\xyoption{all}

\newcommand{\valpha}{\boldsymbol{\alpha}}
\newcommand{\veta}{\boldsymbol{\eta}}

\begin{document}

\title{Revisiting the Equivalence Problem for Finite Multitape Automata}
\author{ James
  Worrell\thanks{Supported by EPSRC grant EP/G069727/1.}}
\institute{Department of Computer Science, University of Oxford, UK}

\maketitle

\begin{abstract}
  The decidability of determining equivalence of deterministic
  multitape automata (or transducers) was a longstanding open problem
  until it was resolved by Harju and Karhum\"{a}ki in the early 1990s.
  Their proof of decidability yields a \textbf{co-NP} upper bound, but
  apparently not much more is known about the complexity of the
  problem.  In this paper we give an alternative proof of
  decidability, which follows the basic strategy of Harju and
  Karhum\"{a}ki but replaces their use of group theory with results on
  matrix algebras.  From our proof we obtain a simple randomised
  algorithm for deciding language equivalence of deterministic
  multitape automata and, more generally, multiplicity equivalence of
  nondeterministic multitape automata.  The algorithm involves only
  matrix exponentiation and runs in polynomial time for each fixed
  number of tapes.  If the two input automata are inequivalent then
  the algorithm outputs a word on which they differ.
\end{abstract}

\section{Introduction}
One-way multitape finite automata were introduced in the seminal 1959
paper of Rabin and Scott~\cite{RabinS59}.  Such automata (under
various restrictions) are also commonly known as transducers---see
Elgot and Mezei~\cite{ElgotM63} for an early reference.  A multitape
automaton with $k$ tapes accepts a $k$-ary relation on words.  The
class of relations recognised by deterministic automata coincides with
the class of $k$-ary rational relations~\cite{ElgotM63}.

Two multitape automata are said to be equivalent if they accept the
same relation.  Undecidability of equivalence of non-deterministic
automata is relatively straightforward~\cite{Griffiths68}.  However
the deterministic case remained open for many years, until it was
shown decidable by Harju and Karhum\"{a}ki~\cite{HarjuK91}.  Their
solution made crucial use of results about ordered
groups---specifically that a free group can be endowed with a
compatible order~\cite{Neumann49b} and that the ring of formal power
series over an ordered group with coefficients in a division ring and
with well-ordered support is itself a division ring (due independently
to Malcev~\cite{Malcev48} and Neumann~\cite{Neumann49}).  Using these
results~\cite{HarjuK91} established the decidability of
\emph{multiplicity equivalence} of non-deterministic multitape
automata, i.e., whether two non-deterministic multitape automata have
the same number of accepting computations on each input.  Decidability
in the deterministic case (and, more generally, the unambiguous case)
follows immediately.  We refer the reader to~\cite{Sakarovich03} for a
self-contained account of the proof, including the underlying group
theory.

Harju and Karhum\"{a}ki did not address questions of complexity
in~\cite{HarjuK91}.  However the existence of a \textbf{co-NP}
\emph{guess-and-check} procedure for deciding equivalence of
deterministic multitape automata follows directly from~\cite[Theorem
  8]{HarjuK91}.  This theorem states that two inequivalent automata
are guaranteed to differ on a tuple of words whose total length is at
most the total number of states of the two automata. Such a tuple can
be guessed, and it can be checked in polynomial time whether the tuple
is accepted by one automaton and rejected by the other.  In the
special case of two-tape deterministic automata, a polynomial-time
algorithm was given in~\cite{FriedmanG82}, before decidability was
shown in the general case.

A \textbf{co-NP} upper bound also holds for multiplicity equivalence
of $k$-tape automata for each fixed $k$.  However, as we observe
below, if the number of tapes is not fixed, computing the number of
accepting computations of a given non-deterministic multitape automata
on a tuple of input words is \#\textbf{P}-hard.  Thus the
guess-and-check method does not yield a \textbf{co-NP} procedure for
multiplicity equivalence in general.

It is well-known that the equivalence problem for single-tape weighted
automata with rational transition weights is solvable in polynomial
time~\cite{Schutzenberger61,Tzeng}.  Now the decision procedure
in~\cite{HarjuK91} reduces multiplicity equivalence of multitape
automata to equivalence of single-tape automata with transition
weights in a division ring of power series over an ordered group.
However the complexity of arithmetic in this ring seems to preclude an
application of the polynomial-time procedures
of~\cite{Schutzenberger61,Tzeng}.  Leaving aside issues of
representing infinite power series, even the operation of multiplying
a family of polynomials in two non-commuting variables yields a result
with exponentially many monomials in the length of its input.

In this paper we give an alternative proof that multiplicity
equivalence of multitape automata is decidable, which also yields new
complexity bounds on the problem.  We use the same basic idea
as~\cite{HarjuK91}---reduce to the single-tape case by enriching the
set of transition weights.  However we replace their use of power
series on ordered groups with results about matrix algebras and
Polynomial Identity rings (see Remark~\ref{rem:difference} for a more
technical comparison).  In particular, we use the Amitsur-Levitzki
theorem concerning polynomial identities in matrix algebras.  Our use
of the latter is inspired by the work of~\cite{BogdanovW05} on
non-commutative polynomial identity testing, and our starting point is
a simple generalisation of the approach of~\cite{BogdanovW05} to what
we call \emph{partially commutative} polynomial identity testing.

Our construction for establishing decidability immediately yields a
simple randomised algorithm for checking multiplicity equivalence of
multitape automata (and hence also equivalence of deterministic
automata).  The algorithm involves only matrix exponentiation, and
runs in polynomial time for each fixed number of tapes. 



\section{Partially Commutative Polynomial Identities}
\label{sec:semi}

\subsection{Matrix Algebras and Polynomial Identities}

Let $F$ be an infinite field.  Recall that an $F$-algebra is a vector
space over $F$ equipped with an associative bilinear product that has
an identity .  Write $F\langle X \rangle$ for the free $F$-algebra
over a set $X$.  The elements of $F\langle X\rangle$ can be viewed as
polynomials over a set of non-commuting variables $X$ with
coefficients in $F$.  Each such polynomial is an $F$-linear
combination of monomials, where each monomial is an element of $X^*$.
The degree of a polynomial is the maximum of the lengths of its
monomials.

Let $A$ be an $F$-algebra and $f \in F\langle X \rangle$.  If $f$
evaluates to $0$ for all valuations of its variables in $A$ then we
say that $A$ satisfies the \emph{polynomial identity} $f=0$.  For
example, an algebra satisfies the polynomial identity $xy-yx=0$ if and
only if it is commutative.  Note that since the variables $x$ and
$y$ do not commute, the polynomial $xy-yx$ is not identically zero.

We denote by $M_n(F)$ the $F$-algebra of $n\times n$ matrices with
coefficients in $F$.  The Amitsur-Levitzki theorem~\cite{AL50,Cohn03}
is a fundamental result about polynomial identities in matrix
algebras.
\begin{theorem}[Amitsur-Levitzki]
The algebra $M_n(F)$ satisfies the polynomial identity
\[ \sum_{\sigma \in S_{2n}} x_{\sigma(1)}\ldots x_{\sigma(2n)} = 0 \, ,\]
where the sum is over the $(2n)!$ elements of the symmetric group $S_{2n}$.
Moreover $M_n(F)$ satisfies no identity of degree less than $2n$.
\label{thm:AL}
\end{theorem}

Given a finite set $X$ of non-commuting variables, the \emph{generic
  $n\times n$ matrix algebra} $F_n\langle X\rangle$ is defined as
follows.  For each variable $x\in X$ we introduce a family of
commuting indeterminates $\{ t^{(x)}_{ij} : 1 \leq i,j \leq n\}$ and
define $F_n\langle X\rangle$ to be the $F$-algebra of $n\times n$
matrices generated by the matrices $(t^{(x)}_{ij})$ for each $x \in X$.
Then $F_n\langle X\rangle$ has the following universal property: any
homomorphism from $F \langle X \rangle$ to a matrix algebra $M_n(R)$,
with $R$ an $F$-algebra, factors uniquely through the map $\Phi^X_n :
F \langle X \rangle \rightarrow F_n \langle X \rangle$ given by
$\Phi^X_n(x)= (t^{(x)}_{ij})$.  

Related to the map $\Phi^X_n$ we also define an $F$-algebra homomorphism
\[ \Psi^X_n : F \langle X \rangle \rightarrow M_n(F\langle t_{ij}^{(x)}
\mid x \in X, 1 \leq i,j \leq n \rangle) \]
by 
 \[ \Psi^X_n(x) = \left( 
\begin{array}{rlllc} 
              0 \; & t_{12}^{(x)}    &        & &\\
                   & \ddots              & \ddots & &\\
                   &                    & & t^{(x)}_{n-1,n} \\
                   &                    & &   0
\end{array}\right) \]
where the matrix on the right has zero entries everywhere but along
the superdiagonal.

\subsection{Partially Commutative Polynomial Identities}
In this section we introduce a notion of \emph{partially commutative
  polynomial identity}.  We first establish notation and recall some
relevant facts about tensor products of algebras.  

Write $A \otimes B$ for the tensor product of $F$-algebras $A$ and $B$,
and write $A^{\otimes k}$ for the $k$-fold tensor power of $A$.  If
$A$ is an algebra of $a\times a$ matrices and $B$ an algebra of
$b\times b$ matrices, then we identify the tensor product $A\otimes B$
with the algebra of $ab\times ab$ matrices spanned by the
matrices $M \otimes N$, $M=(m_{ij}) \in A$ and
$N=(n_{ij}) \in B$, where 
\[ M \otimes N = \left(
 \begin{array}{ccc}
  m_{11}N & \cdots & m_{1a}N \\ \vdots & & \vdots \\ m_{a1}N & \cdots
  & m_{aa} N
 \end{array}\right ) \]
In particular we have $F^{\otimes k}=F$.

A \emph{partially commuting set of variables} is a tuple
$\boldsymbol{X}=(X_1,\ldots,X_k)$, where the $X_i$ are disjoint sets.
Write $F \langle \boldsymbol{X} \rangle$ for the tensor product
$F\langle X_1\rangle \otimes \cdots \otimes F\langle X_k\rangle$.  We
think of $F \langle \boldsymbol{X} \rangle$ as a set of polynomials in
partially commuting variables.  Intuitively two variables $x, y \in
X_i$ do not commute, whereas $x\in X_i$ commutes with $y \in X_j$ if
$i\neq j$.  Note that if each $X_i$ is a singleton $\{x_i\}$ then $F
\langle \boldsymbol{X} \rangle$ is the familiar ring of polynomials in
commuting variables $x_1,\ldots,x_k$.  At the other extreme, if $k=1$
then we recover the non-commutative case.

An arbitrary element $f \in F\langle \boldsymbol{X} \rangle = F\langle
X_1\rangle \otimes \cdots \otimes F\langle X_k\rangle$ can be written
uniquely as a finite sum of distinct \emph{monomials}, where each
monomial is a tensor product of elements of $X_1^*$, $ X_2^*$,\ldots,
and $X_k^*$.  Formally, we can write
\begin{gather}
f = \sum_{i \in I} \alpha_i(m_{i,1} \otimes \cdots \otimes m_{i,k}) \, ,
\label{eq:normalform}
\end{gather}
where $\alpha_i \in F$ and $m_{i,j} \in X_j^*$ for each $i\in I$ and
$1\leq j \leq k$.  Thus we can identify $F \langle \boldsymbol{X}
\rangle$ with the free $F$-algebra over the product monoid $X_1^* \times
\ldots \times X_k^*$.

Define the \emph{degree} of a monomial $m_1 \otimes \ldots \otimes
m_k$ to be the total length $|m_1|+\ldots+|m_k|$ of its constituent
words.  The degree of a polynomial is the maximum of the degrees of
its constituent monomials.

Let $\boldsymbol{A}=(A_1,\ldots,A_k)$ be a $k$-tuple of $F$-algebras.
A \emph{valuation} of $F \langle \boldsymbol{X} \rangle$ in
$\boldsymbol{A}$ is a tuple of functions
$\boldsymbol{v}=(v_1,\ldots,v_k)$, where $v_i : X_i \rightarrow A_i$.
Each $v_i$ extends uniquely to an $F$-algebra homomorphism
$\widetilde{v_i} : F\langle X_i \rangle \rightarrow A_i$, and we
define the map $\widetilde{\boldsymbol{v}} : F \langle \boldsymbol{X}
\rangle \rightarrow A_1 \otimes \ldots \otimes A_k$ by
$\widetilde{v}=\widetilde{v_1} \otimes \ldots \otimes
\widetilde{v_k}$.  Often we will abuse terminology slightly and speak
of a valuation of $F \langle \boldsymbol{X} \rangle$ in $A_1 \otimes
\cdots \otimes A_k$.  Given $f \in F \langle \boldsymbol{X} \rangle$,
we say that $\boldsymbol{A}$ satisfies the partially commutative
identity $f=0$ if $\widetilde{\boldsymbol{v}}(f)=0$ for all valuations
$\boldsymbol{v}$.

Next we introduce two valuations that will play an important role in
the subsequent development.  Recall that given a set of non-commuting
variables $X$, we have a map $\Phi^X_n : F \langle X \rangle
\rightarrow F_n \langle X \rangle$ from the free $F$-algebra to the
generic $n$-dimensional matrix algebra.  We now define a valuation
\begin{gather}
\Phi^{\boldsymbol{X}}_n : F\langle \boldsymbol{X} \rangle \longrightarrow
F_n\langle X_1\rangle \otimes \cdots \otimes F_n\langle X_k\rangle 
\label{eq:psi}
\end{gather}
by $\Phi^{\boldsymbol{X}}_n = \Phi^{X_1}_{n} \otimes \cdots \otimes
\Phi_{n}^{X_k}$.    Likewise we
define
\begin{align*} \Psi^{\boldsymbol{X}}_n : F \langle \boldsymbol{X} \rangle
\longrightarrow
&  M_n(F\langle t_{ij}^{(x)}
\mid x \in X_1, 1 \leq i,j \leq n \rangle)  \otimes \cdots \\
& \otimes 
 M_n(F\langle t_{ij}^{(x)}
\mid x \in X_k, 1 \leq i,j \leq n \rangle)  
\end{align*}
by $\Psi_n^{\boldsymbol{X}} = \Psi_n^{X_1} \otimes \cdots \otimes \Psi_n^{X_k}$.
We will usually elide the superscript $\boldsymbol{X}$ from 
$\Phi^{\boldsymbol{X}}_n$ and  $\Psi_n^{\boldsymbol{X}}$ when it is clear from the 
context.

The following result generalises (part of) the Amitsur-Levitzki
theorem, by giving a lower bound on the degrees of partially
polynomial identities holding in tensor products of matrix algebras.

\begin{proposition}
Let $f \in F \langle \boldsymbol{X} \rangle$ and let $L$ be a field
extending $F$.  Then the following are equivalent: (i)~The partially
commutative identity $f=0$ holds in $M_n(L) \otimes_F \cdots_F \otimes
M_n(L)$; (ii)~$\Phi_n(f) = 0$.  Moreover, if $f$ has degree strictly
less than $n$ then (i) and (ii) are both equivalent to
(iii)~$\Psi_n(f) = 0$; and (iv)~$f$ is identically $0$ in $F \langle
\boldsymbol{X} \rangle$.
\label{prop:sc-AL}
\end{proposition}
\begin{proof}
  The implication (ii) $\Rightarrow$ (i) follows from the fact that
  any valuation from $F \langle \boldsymbol{X} \rangle$ to $M_n(L)
  \otimes_F \cdots_F \otimes M_n(L)$ factors through $\Phi_n$.  To
  see that (i) $\Rightarrow$ (ii), observe that $\Phi_n(f)$ is an
  $n^k\times n^k$ matrix in which each entry is a polynomial in the
  commuting variables $t^{(x)}_{ij}$.  Condition (i) implies in particular
  that each such polynomial evaluates to $0$ for all valuations of its
  variables in $F$.  Since $F$ is an infinite field, it must be that
  each such polynomial is identically zero, i.e., (ii) holds.

  The implications (ii) $\Rightarrow$ (iii) and (iv) $\Rightarrow$ (i)
  are both straightforward, even without the degree restriction on $f$.

Finally we show that (iii) $\Rightarrow$ (iv).  Let $m_1 \otimes
\ldots \otimes m_k$ be a monomial in $F \langle \boldsymbol{X}
\rangle$, where $m_i = m_{i,1} \ldots m_{i,l_i} \in X_i^*$ has length
$l_i < n$.  Then $\Psi_n(m_1 \otimes \cdots \otimes m_k)$ is an $n^k
\times n^k$ matrix whose first row has a single non-zero entry, which
is the monomial
\begin{gather}
t_{12}^{(m_{1,1})}\ldots t_{l_1,l_1+1}^{(m_{1,l_1})} \ldots
t_{12}^{(m_{k,1})}\ldots t_{l_k,l_k+1}^{(m_{k,l_k})} 
\label{eq:psi2}
\end{gather}
at index $(1,\ldots,1),(l_1+1,\ldots,l_k+1)$.

It follows that $\Psi_n$ maps the set of monomials in $F \langle
\boldsymbol{X} \rangle$ of degree less than $n$ injectively into a 
linearly independent set of matrices.  Condition (iv)
immediately follows.  \qed
\end{proof}
The hypothesis that $f$ have degree less than $n$ in
Proposition~\ref{prop:sc-AL} can be weakened somewhat, but is
sufficient for our purposes.

\subsection{Division Rings and Ore Domains}
A ring $R$ with no zero divisors is a \emph{domain}.  If moreover each
non-zero element of $R$ has a two-sided multiplicative inverse, then we
say that $R$ is a \emph{division ring} (also called a
\emph{skew field}).  A domain $R$ is a \emph{(right) Ore domain}
if for all $a,b \in R \setminus \{0\}$, $aR \cap bR \neq 0$.  The
significance of this notion is that an Ore domain can be embedded in a
division ring of fractions~\cite[Corollary 7.1.6]{Cohn03}, something
that need not hold for an arbitrary domain.  If the Ore condition
fails then it can easily be shown that the subalgebra of $R$ generated
by $a$ and $b$ is free on $a$ and $b$.  It follows that a domain $R$
that satisfies some polynomial identity is an Ore
domain~\cite[Corollary 7.5.2]{Cohn03}.

\begin{proposition}
  The tensor product of generic matrix algebras
  $F_n\langle X_1\rangle \otimes \cdots \otimes F_n\langle X_k\rangle$
  is an Ore domain for each $n \in \mathbb{N}$.
\label{prop:ore}
\end{proposition}
\begin{proof}[sketch]
  We give a proof sketch here, deferring the details to
  Appendix~\ref{sec:appendix}.

By the Amitsur-Levitzki theorem, $F_n\langle X_1 \rangle \otimes
\cdots \otimes F_n \langle X_n \rangle$ satisfies a polynomial
identity.  Thus it suffices to show that $F_n\langle X_1 \rangle
\otimes \cdots \otimes F_n \langle X_n \rangle$ is a domain for each
$n$.  Now it is shown in~\cite[Proposition 7.7.2]{Cohn03} that
$F_n\langle X\rangle$ is a domain for each $n$ and set of variables
$X$.  While the tensor product of domains need not be a domain (e.g.,
$\mathbb{C} \otimes_{\mathbb{R}} \mathbb{C} \cong \mathbb{C} \times
\mathbb{C}$), the proof in~\cite{Cohn03} can be adapted \emph{mutatis
  mutandis} to show that $F \langle X_1 \rangle \otimes \cdots \otimes
F \langle X_k \rangle$ is also a domain.

To prove the latter, it suffices to find central simple $F$-algebras
$D_1,\ldots,D_k$, each of degree $n$, such that the $k$-fold tensor
product $D\otimes_F \cdots \otimes_F D$ is a domain.  Such an example
can be found, e.g., in~\cite[Proposition 1.1]{Saltman}.  Then, using
the fact that $D \otimes_F L \cong M_n(L)$ for any algebraically
closed extension field of $F$, one can infer that $F_n \langle X_1
\rangle \otimes \cdots \otimes F_n \langle X_k \rangle$ is also a
domain.  \qed
\end{proof}

\section{Multitape Automata}
\label{sec:multitape}
Let $\Sigma=(\Sigma_1,\ldots,\Sigma_k)$ be a tuple of finite
alphabets.  We denote by $S$ the product monoid $\Sigma_1^* \times
\cdots \times \Sigma_k^*$.  Define the length of $s=(w_1,\ldots,w_k)
\in S$ to be $|s|=|w_1|+\ldots+|w_k|$ and write $S^{(l)}$ for the set
of elements of $S$ of length $l$.  A \emph{multitape automaton} is a
tuple $A=(\Sigma,Q,E,Q_0,F)$, where $Q$ is a set of \emph{states}, $E
\subseteq Q\times S^{(1)} \times Q$ is a set of \emph{edges}, $Q_0
\subseteq Q$ is a set of \emph{initial states}, and $Q_f \subseteq Q$ is
a set of \emph{final states}.  A \emph{run} of $A$ from state $q_0$ to
state $q_m$ is a finite sequence of edges $\rho=e_1e_2 \ldots e_m$
such that $e_i = (q_{i-1},s_i,q_i)$.  The \emph{label} of $\rho$ is
the product $s_1s_2\ldots s_m \in S$.  Define the \emph{multiplicity}
$A(s)$ of an input $s \in S$ to be the number of runs with label $s$
such that $q_0\in Q_0$ and $q_m \in Q_f$.  An automaton is
\emph{deterministic} if each state reads letters from a single tape
and has a single transition for every input letter.  Thus a
deterministic automaton has a single run on each input $s\in S$.

\subsection{Multiplicity Equivalence}
We say that two automata $A$ and $B$ over the same alphabet are
\emph{multiplicity equivalent} if $A(s)=B(s)$ for all $s \in S$.  The
following result implies that multiplicity equivalence of multitape
automata is decidable.

\begin{theorem}[Harju and Karhum\"{a}ki]
Given automata $A$ and $B$ with $n$ states in total, $A$
and $B$ are equivalent if and only if $A(s)=B(s)$ for all $s\in S$ of
length at most $n-1$.
\label{thm:harju}
\end{theorem}

Theorem~\ref{thm:harju} immediately yields a \textbf{co-NP} bound for
checking language equivalence of deterministic multitape automata.
Given two inequivalent automata $A$ and $B$, a distinguishing input
$s$ can be guessed, and it can be verified in polynomial time that
only one of $A$ and $B$ accepts $s$.  A similar idea also gives a
\textbf{co-NP} bound for multiplicity equivalence in case the number
of tapes is fixed.  In general we note however that the
\emph{evaluation problem}---given an automaton $A$ and input $s$,
compute $A(s)$---is $\# \mathbf{P}$-complete.  Thus it is not clear
that the \textbf{co-NP} upper bound applies to the multiplicity
equivalence problem without bounding the number of tapes.

\begin{proposition}
The evaluation problem for multitape automata is $\#
\mathbf{P}$-complete.
\end{proposition}
\begin{proof} 
  Membership in $\# \mathbf{P}$ follows from the observation that
  a non-deterministic polynomial-time algorithm can enumerate all
  possible runs of an automaton $A$ on an input $s\in S$.

  The proof of $\# \mathbf{P}$-hardness is by reduction from \#SAT, the
  problem of counting the number of satisfying assignments of a
  propositional formula.  Consider such a formula $\varphi$ with $k$
  variables, each with fewer than $n$ occurrences.  We define a
  $k$-tape automaton $A$, with each tape having alphabet $\{0,1\}$,
  and consider as input the $k$-tuple $s=((01)^n,\ldots,(01)^n)$.  The
  automaton $A$ is constructed such that its runs on input $s$ are in
  one-to-one correspondence with satisfying assignments of $\varphi$.
  Each run starts with the automaton reading the symbol $0$ from a
  non-deterministically chosen subset of its tapes, corresponding to
  the set of false variables.  Thereafter it evaluates the formula
  $\varphi$ by repeatedly guessing truth values of the propositional
  variables.  If the $i$-th variable is guessed to be true then the
  automaton reads $01$ from the $i$-th tape; otherwise it reads $10$
  from the $i$-th tape.  The final step is to read the symbol $1$ from
  a non-deterministically chosen subset of the input tapes---again
  corresponding to the set of false variables.  The consistency of the
  guesses is ensured by the requirement that the automaton have read
  $s$ by the end of the computation.\qed
\end{proof}

\subsection{Decidability}
We start by recalling from~\cite{HarjuK91} an equivalence-respecting
transformation from multitape automata to single-tape weighted
automata.  

Recall that a single-tape automaton on a unary alphabet with
transition weights in a ring $R$ consists of a set of \emph{states}
$Q=\{q_1,\ldots,q_n\}$, \emph{initial states} $Q_0\subseteq Q$,
\emph{final states} $Q_f\subseteq Q$, and \emph{transition matrix} $M
\in M_n(R)$.  Given such an automaton, define the
\emph{initial-state vector} $\valpha \in R^{1\times n}$ and
\emph{final-state vector} $\veta \in R^{n\times 1}$ respectively by
\[ \alpha_i = \left\{  \begin{array}{ll}
                    1 & \mbox{ if $q_i \in Q_0$}\\
                    0 & \mbox{ otherwise}
\end{array} \right . \qquad \mbox{and} \qquad 
   \eta_i = \left\{  \begin{array}{ll}
                    1 & \mbox{ if $q_i \in Q_f$} \\
                    0 & \mbox{ otherwise}
\end{array}\right . \]
Then $\valpha M^l \veta$ is the weight of the (unique) input word of
length $l$.

Consider a $k$-tape automaton $A=(\boldsymbol{\Sigma},Q,E,Q_0,Q_f)$,
where $\boldsymbol{\Sigma} = (\Sigma_1,\ldots,\Sigma_k)$, and write
$S=\Sigma_1^*\times \cdots \times \Sigma_k^*$.  Recall the ring of
polynomials \[ F \langle \boldsymbol{\Sigma} \rangle = F\langle
\Sigma_1\rangle \otimes \cdots \otimes F\langle \Sigma_n\rangle,\] as
defined in Section~\ref{sec:semi}.  Recall also that we can identify
the monoid $S$ with the set of monomials in $F \langle
\boldsymbol{\Sigma} \rangle$, where $(w_1,\ldots,w_k) \in S$ corresponds to
$w_1 \otimes \cdots \otimes w_k$---indeed $F \langle
\boldsymbol{\Sigma} \rangle$ is the free $F$-algebra on $S$.

We derive from $A$ an $F \langle \boldsymbol{\Sigma} \rangle$-weighted
automaton $\widetilde{A}$ (with a single tape and unary input
alphabet) that has the same sets of states, initial states, and final
states as $A$.  We define the transition matrix $M$ of $\widetilde{A}$ by
combining the different transitions of $A$ into a single matrix with
entries in $F \langle \boldsymbol{\Sigma} \rangle$.  To this
end, suppose that the set of states of $A$ is $Q=\{q_1,\ldots,q_n\}$.
Define the matrix $M \in M_n(F \langle \boldsymbol{\Sigma} \rangle)$
by $M_{ij} = \sum_{(q_i,s,q_j)\in E} s$ for $1 \leq i,j \leq n$.

Let $\valpha$ and $\veta$ be the respective initial- and final-state
vectors of $\widetilde{A}$.  Then the following proposition is
straightforward.  Intuitively it says that the weight of the unary
word of length $l$ in $\widetilde{A}$ represents the language of all
length-$l$ tuples accepted by $A$.
\begin{proposition}
For all $l \in \mathbb{N}$ we have
$\valpha M^l \veta = \sum_{s\in S^{(l)}} A(s) \cdot s$.
\label{prop:matrix}
\end{proposition}

Now consider two $k$-tape automata $A$ and $B$.  Let the weighted
single-tape automata derived from $A$ and $B$ have respective
transition matrices $M_A$ and $M_B$, initial-state vectors $\valpha_A$
and $\valpha_B$, and final-state vectors $\veta_A$ and $\veta_B$.  We
combine the latter into a single weighted automaton with transition
matrix $M$, initial-state vector $\valpha$, and final-state vector
$\veta$, respectively defined by:
\[ \valpha = ( \valpha_A \; \valpha_B) \qquad
   M = \left(\begin{array}{cc} M_A & 0\\ 0 & M_B 
     \end{array} \right) \qquad
   \veta = \left(\begin{array}{c} \veta_A \\ -\veta_B 
     \end{array} \right) \qquad
\]
\begin{proposition}
  Automata $A$ and $B$ are multiplicity equivalent if and only if
  $\valpha M^l \veta = 0$ for $l=0,1,\ldots,n-1$, where $n$ is the
  total number of states of the two automata.
\label{prop:main}
\end{proposition}
\begin{proof}
Since $S$ is a linearly independent subset of $F \langle
\boldsymbol{\Sigma} \rangle$, from Proposition~\ref{prop:matrix} it
follows that $A$ and $B$ are multiplicity equivalent just in case
$\valpha_A (M_A)^l \veta_A = \valpha_B (M_B)^l \veta_B$ for all $l \in
\mathbb{N}$.  The latter is clearly equivalent to $\valpha M^l \veta =
0$ for all $l \in \mathbb{N}$.  It remains to show that we can check
equivalence by looking only at exponents $l$ in the range
$0,1,\ldots,n-1$.

Suppose that $\valpha M^i \veta = 0$ for $i=0,\ldots,n-1$.  We show
that $\valpha M^l \veta = 0$ for an arbitrary $l \geq n$.

Consider the map $\Phi_l :F \langle
\boldsymbol{\Sigma} \rangle\rightarrow F_l\langle
\Sigma_1\rangle\otimes\cdots \otimes F_l\langle \Sigma_k\rangle$, as
defined in (\ref{eq:psi}).  Observe that $\valpha M^l \veta$ is a
polynomial expression in $F \langle \boldsymbol{\Sigma} \rangle$ of degree at
most $l$.  Therefore by Proposition~\ref{prop:sc-AL} ((ii)
$\Leftrightarrow$ (iv)), to show that $\valpha M^l \veta = 0$ it
suffices to show that
\begin{gather}
\Phi_l (\valpha M^l \veta ) \, = \, 0 \, .
\label{eq:goal4}
\end{gather}
Let us write $\Phi_l(M)$ for the pointwise application of $\Phi_l$
to the matrix $M$, so that $\Phi_l(M)$
is an $n\times n$ matrix, each of whose entries is an $n^k\times n^k$
matrix belonging to
$F_l\langle \Sigma_1\rangle\otimes\cdots
\otimes F_l\langle \Sigma_k\rangle$.  Since $\Phi_l$ is a homomorphism
and $\valpha$ and $\veta$ are integer vectors,
(\ref{eq:goal4}) is equivalent to:
\begin{gather}
\valpha \, \Phi_l(M)^l \,\veta \, = \, 0 \, .
\label{eq:goal2}
\end{gather}

Recall from Proposition~\ref{prop:ore} that the tensor product of
generic matrix algebras $F_l\langle \Sigma_1\rangle\otimes\cdots
\otimes F_l\langle \Sigma_k\rangle$ is an Ore domain and hence can be
embedded in a division ring.  Now a standard result about single-tape
weighted automata with transition weights in a division ring is that
such an automaton with $n$ states is equivalent to the zero automaton
if and only if it assigns zero weight to all words of length $n$
(see~\cite[pp143--145]{Eilenberg74} and~\cite{Schutzenberger61}).
Applying this result to the unary weighted automaton defined by
$\valpha$, $M$, and $\veta$, we see that (\ref{eq:goal2}) is implied
by
\begin{gather}
\valpha \,\Phi_l(M)^i\, \veta = \, 0
\qquad i=0,1,\ldots,n-1 \, .
\label{eq:goal3}
\end{gather}
But, since $\Phi_l$ is a homomorphism, (\ref{eq:goal3}) is implied by 
\begin{gather}
\valpha M^i \veta = 0 \qquad i=0,1,\ldots,n-1 \, .
\end{gather}
This concludes the proof.\qed
\end{proof}
Theorem~\ref{thm:harju} immediately follows from
Proposition~\ref{prop:main}.  

\begin{remark}
  The difference between our proof of Theorem~\ref{thm:harju} and the
  proof in~\cite{HarjuK91} is that we consider a family of
  homomorphisms of $F \langle \boldsymbol{\Sigma} \rangle$ into Ore
  domains of matrices---the maps $\Phi_l$---rather than a single
  ``global'' embedding of $F \langle \boldsymbol{\Sigma} \rangle$ into
  a division ring of power series over a product of free groups.  None
  of the maps $\Phi_l$ is an embedding, but it suffices to use the
  lower bound on the degrees of polynomial identities in
  Proposition~\ref{prop:sc-AL} in lieu of injectivity.  On the other
  hand, the fact that $F_l \langle \Sigma_1 \rangle \otimes \cdots
  \otimes F_l \langle \Sigma_k \rangle$ satisfies a polynomial identity
  makes it relatively straightforward to exhibit an embedding of the
  latter into a division ring.  As we now show, this approach leads
  directly to a very simple randomised polynomial-time algorithm for
  solving the equivalence problem.
\label{rem:difference}
\end{remark}

\subsection{Randomised Algorithm}

Proposition~\ref{prop:main} reduces the problem of checking
multiplicity equivalence of multitape automata $A$ and $B$ to checking
the partially commutative identities $\valpha M^l \veta = 0$, $l
=0,1,\ldots,n-1$ in $F \langle \boldsymbol{\Sigma} \rangle$.  Since
each identity has degree less than $n$, applying
Proposition~\ref{prop:sc-AL} ((iii) $\Leftrightarrow$ (iv)) we see
that $A$ and $B$ are equivalent if and only if
\begin{gather}
\valpha \, \Psi_n(M)^l \, \veta = 0 
\qquad
l=0,1,\ldots,n-1 \, .
\label{eq:pit}
\end{gather}

Each equation $\valpha \Psi_n(M)^l \veta=0$ in (\ref{eq:pit})
asserts the zeroness of an $n^k\times n^k$ matrix of polynomials in
the commuting variables $t^{(x)}_{ij}$, with each polynomial having degree
less than $n$.  Suppose that $\valpha  \Psi_n(M)^l \veta \neq 0$
for some $l$---say the matrix entry with index
$((1,\ldots,1),(l_1+1,\ldots,l_k+1))$ contains a monomial with non-zero
coefficient.  By (\ref{eq:psi2}) such a monomial determines a term $s
\in \Sigma_1^{l_1} \times \cdots \times \Sigma_k^{l_k}$ with non-zero
coefficient in $\valpha M^l \veta$, and by
Proposition~\ref{prop:matrix} we have $A(s)\neq B(s)$.

We can verify each polynomial identity in (\ref{eq:pit}), outputting a
monomial of any non-zero polynomial, using a classical identity testing
procedure based on the \emph{isolation lemma} of~\cite{MVV87}.

\begin{lemma}[\cite{MVV87}]
  There is a randomised polynomial-time algorithm that inputs a
  multilinear polynomial $f(x_1,\ldots,x_m)$, represented as an
  algebraic circuit, and either outputs a monomial of $f$ or that $f$
  is zero.  Moreover the algorithm is always correct if $f$ is the
  zero polynomial and is correct with probability at least $1/2$ if
  $f$ is non-zero.
\label{lem:pit}
\end{lemma}

The idea behind the algorithm described in Lemma~\ref{lem:pit} is to
choose a weight $w_i \in \{1,\ldots,2m\}$ for each variable $x_i$ of
$f$ independently and uniformly at random.  Defining the weight of a
monomial $x_{i_1}\ldots x_{i_t}$ to be $w_{i_1}+\ldots+w_{i_t}$, then
with probability at least $1/2$ there is a unique minimum-weight
monomial.  The existence of a minimum-weight monomial can be detected
by computing the polynomial $g(y)=f(y^{w_1},\ldots,y^{w_k})$, since a
monomial with weight $w$ in $f$ yields a monomial of degree $w$ in
$g$.  Using similar ideas one can moreover determine the composition of
a minimum-weight monomial in $f$.

Applying Lemma~\ref{lem:pit} we obtain our main result:
\begin{theorem}
Let $k$ be fixed.  Then multiplicity equivalence of $k$-tape automata can
be decided in randomised polynomial time.  Moreover there is a
randomised polynomial algorithm for the function problem of computing
a distinguishing input given two inequivalent automata.
\label{thm:main}
\end{theorem}

The reason for the requirement that $k$ be fixed is because the
dimension of the entries of the transition matrix $M$, and thus the
number of polynomials to be checked for equality, depends
exponentially on $k$.

The above use of the isolation technique
generalises~\cite{KieferMOWW13}, where it is used to generate
counterexample words of weighted single-tape automata.  A very similar
application in~\cite{ArvindM08} occurs in the context of identity
testing for non-commutative algebraic branching programs.

\section{Conclusion}
We have given a simple randomised algorithm for deciding language
equivalence of deterministic multitape automata and multiplicity
equivalence of nondeterministic automata.  The algorithm arises
directly from algebraic constructions used to establish decidability
of the problem, and runs in polynomial time for each fixed number of
tapes.  We leave open the question of whether there is a deterministic
polynomial-time algorithm for deciding the equivalence of
deterministic and weighted multitape automata with a fixed number of
tapes.  (Recall that the 2-tape case is already known to be in
polynomial time~\cite{FriedmanG82}.)  We also leave open whether there
is a deterministic or randomised polynomial time algorithm for solving
the problem in case the number of tapes is not fixed.



\appendix
\section{Proof of Proposition~\ref{prop:ore}}
\label{sec:appendix}
We first recall a construction of a \emph{crossed product division
  algebra} from~\cite[Proposition 1.1]{Saltman}.  Let $z_1,\ldots,z_k$
be commuting indeterminates and write
$F=\mathbb{Q}(z^n_1,\ldots,z^n_k)$ for the field of rational functions
obtained by adjoining $z^n_1,\ldots,z^n_k$ to $\mathbb{Q}$.
Furthermore, let $K/F$ be a field extension whose Galois group is
generated by commuting automorphisms $\sigma_1,\ldots,\sigma_k$, each
of order $n$, which has fixed field $F$.  (Such an extension can easily
be constructed by adjoining extra indeterminates to $F$, and having
the $\sigma_i$ be suitable permutations of the new indeterminates.)
For each $i$, $1 \leq i \leq k$, write $K_i$ for the subfield of $K$
that is fixed by each $\sigma_j$ for $j\neq i$; then define $D_i$ to
be the $F$-algebra generated by $K_i$ and $z_i$ such that $az_i =
z_i\sigma_i(a)$ for all $a \in K_i$.  Then each $D_i$ is a simple
algebra of dimension $n^2$ over its centre $F$.  It is shown
in~\cite[Proposition 1.1]{Saltman} that the tensor product $D_1
\otimes_F \cdots \otimes_F D_k$ can be characterised as the
localisation of an iterated skew polynomial ring---and is therefore a
domain.

The following two propositions are straightforward adaptations
of~\cite[Proposition 7.5.5.]{Cohn03} and~\cite[Proposition
  7.7.2]{Cohn03} to partially commutative identities.
\begin{proposition}
Let $f \in F\langle X_1\rangle \otimes \cdots \otimes F\langle
X_k\rangle$.  If the partially commutative identity $f=0$ holds in
$D_1 \otimes_F \cdots \otimes_F D_k$ then it also holds in $(D_1
\otimes_F L) \otimes_F \cdots \otimes_F (D_k \otimes_F L)$ for any
extension field $L$ of $F$.
\label{prop:extend}
\end{proposition}
\begin{proof}
Noting that the $D_i$ are all isomorphic as $F$-algebras, let
$\{e_1,\ldots,e_{n^2}\}$ be a basis of each $D_i$ over its centre $F$.
For each variable $x$ appearing in $f$, introduce commuting
indeterminates $t_{xj}$, $1 \leq j \leq n^2$, and write
$x=\sum_{j=1}^{n^2} t_{xj}e_j$.  Then we can express $f$ in the form
\begin{gather} f = \sum_{\nu \in \{1,\ldots,n^2\}^{k}} f_\nu \cdot 
(e_{\nu(1)} \otimes \cdots \otimes e_{\nu(k)}) \, , 
\label{eq:formula}
\end{gather}
where $f_\nu \in 
F \langle t_{xj} : x\in X_1,1 \leq j\leq n^2 \rangle
\otimes_F \cdots \otimes_F F 
\langle t_{xj} : x\in X_k,1 \leq j\leq n^2 \rangle$.

By assumption, each $f_\nu$ evaluates to $0$ for all values of the
$t_{xj}$ in $F$.  Since $F$ is an infinite field it follows that each
$f_\nu$ must be identically zero.  Now we can also regard
$\{e_1,\ldots,e_{n^2}\}$ as a basis for $D_i \otimes_F L$ over $L$.
Then by (\ref{eq:formula}), $f=0$ also on $(D_1 \otimes_F L) \otimes_F
\cdots \otimes_F (D_k \otimes_F L)$.  \qed
\end{proof}

\begin{proposition}
$F_n\langle X_1\rangle \otimes \cdots \otimes F_n\langle X_k\rangle$
  is a domain.
\label{prop:domain}
\end{proposition}
\begin{proof}
Recall that if $L$ is an algebraically closed field extension of $F$,
then we have $D_i \otimes_F L \cong M_n(L)$ for each $i$.  By
Proposition~\ref{prop:extend} it follows that an identity $f=0$ holds
in $D_1 \otimes_F \cdots \otimes_F D_k$ if and only if it holds in $M_n(L)
\otimes_F \cdots \otimes_F M_n(L)$.  But by
Proposition~\ref{prop:sc-AL} the latter holds if and only if
$\Phi_n(f)$ is identically zero.

To prove the proposition it will suffice to show that the image of
$\Phi_n$ contains no zero divisors, since the latter is a surjective map.
Now given $f,g \in F\langle X_1\rangle \otimes \cdots \otimes F\langle
X_k\rangle$ with $\Phi_n(fg)=0$, we have that $D_1 \otimes_F \cdots
\otimes_F D_k$ satisfies the identity $fg=0$.  Since $D_1 \otimes_F \cdots
\otimes_F D_k$ is a domain, it follows that it satisfies the identity
$fhg=0$ for any $h$ in $F\langle X_1\rangle \otimes \cdots \otimes
F\langle X_k\rangle$.  But now $M_n(L) \otimes_F \cdots \otimes_F
M_n(L)$ satisfies the identity $fhg=0$ for any $h$.  Since $h$ can
take the value of an arbitrary matrix (in particular, any matrix unit)
it follows that $M_n(L) \otimes_F \cdots \otimes_F M_n(L)$ satisfies
either the identity $f=0$ or $g=0$, and so, by
Proposition~\ref{prop:sc-AL} again, either $\Phi_n(f)=0$ or
$\Phi_n(g)=0$.  \qed
\end{proof}

\subsubsection{Acknowledgments}
The author is grateful to Louis Rowen for helpful pointers in the
proof of Proposition~\ref{prop:ore}.

\bibliographystyle{plain} 

\begin{thebibliography}{10}

\bibitem{AL50}
S.A.~Amitsur and J.~Levitzki.
\newblock Minimal identities for algebras.
\newblock {\em Proceedings of the American Mathematical Society}, 1:449--463,
  1950.

\bibitem{ArvindM08}
V.~Arvind and P.~Mukhopadhyay.
\newblock Derandomizing the isolation lemma and lower bounds for circuit size.
\newblock In {\em APPROX-RANDOM}, volume 5171 of {\em Lecture Notes in Computer
  Science}, pages 276--289. Springer, 2008.

\bibitem{BogdanovW05}
A.~Bogdanov and H.~Wee.
\newblock More on noncommutative polynomial identity testing.
\newblock In {\em IEEE Conference on Computational Complexity}, pages 92--99.
  IEEE Computer Society, 2005.

\bibitem{Cohn03}
P.M.~Cohn.
\newblock {\em Further Algebra and Applications}.
\newblock Springer-Verlag, 2003.

\bibitem{Eilenberg74}
S.~Eilenberg.
\newblock {\em Automata, Languages, and Machines, Vol.A}.
\newblock Academic Press, 1974.

\bibitem{ElgotM63}
C.C.~Elgot and J.E.~Mezei.
\newblock Two-sided finite-state transductions (abbreviated version).
\newblock In {\em SWCT (FOCS)}, pages 17--22. IEEE Computer Society, 1963.

\bibitem{FriedmanG82}
E.~P. Friedman and S.~A. Greibach.
\newblock A polynomial time algorithm for deciding the equivalence problem for
  2-tape deterministic finite state acceptors.
\newblock {\em SIAM J. Comput.}, 11(1):166--183, 1982.

\bibitem{Griffiths68}
T.V.~ Griffiths.
\newblock The unsolvability of the equivalence problem for $\epsilon$-free
  nondeterministic generalized machines.
\newblock {\em J. ACM}, 15(3):409--413, July 1968.

\bibitem{HarjuK91}
T.~Harju and J.~Karhum{\"a}ki.
\newblock The equivalence problem of multitape finite automata.
\newblock {\em Theor. Comput. Sci.}, 78(2):347--355, 1991.

\bibitem{KieferMOWW13}
S.~Kiefer, A.~Murawski, J.~Ouaknine, B.~Wachter, and J.~Worrell.
\newblock On the complexity of equivalence and minimisation for {Q}-weighted
  automata.
\newblock {\em Logical Methods in Computer Science}, 9, 2013.

\bibitem{Malcev48}
A.I.~ Malcev.
\newblock On the embedding of group algebras in division algebras.
\newblock {\em Dokl. Akad. Nauk}, 60:1409--1501, 1948.

\bibitem{MVV87}
K.~Mulmuley, U.V.~Vazirani, and V.V.~Vazirani.
\newblock Matching is as easy as matrix inversion.
\newblock In {\em STOC}, pages 345--354, 1987.

\bibitem{Neumann49b}
B.H.~Neumann.
\newblock On ordered groups.
\newblock {\em Amer. J. Math.}, 71:1--18, 1949.

\bibitem{Neumann49}
B.H.~Neumann.
\newblock On ordered division rings.
\newblock {\em Trans. Amer. Math. Soc.}, 66:202--252, 1949.

\bibitem{RabinS59}
M.~Rabin and D.~Scott.
\newblock Finite automata and their decision problems.
\newblock {\em IBM Journal of Research and Development}, 3(2):114--125, 1959.

\bibitem{Sakarovich03}
J.~Sakarovich.
\newblock {\em Elements of Automata Theory}.
\newblock Cambridge University Press, 2003.

\bibitem{Saltman}
D.~Saltman.
\newblock {\em Lectures on Division Algebras}.
\newblock American Math. Soc., 1999.

\bibitem{Schutzenberger61}
M.-P.~Sch\"{u}tzenberger.
\newblock On the definition of a family of automata.
\newblock {\em Inf.\ and Control}, 4:245--270, 1961.

\bibitem{Tzeng}
W.~Tzeng.
\newblock A polynomial-time algorithm for the equivalence of probabilistic
  automata.
\newblock {\em SIAM Journal on Computing}, 21(2):216--227, 1992.

\end{thebibliography}

\end{document}